\documentclass{ws-ijfcs}

\usepackage{amsmath}
\usepackage{amssymb}
\usepackage{array}
\usepackage{complexity}
\usepackage{endnotes}
\usepackage{enumitem}
\usepackage{etoolbox}
\usepackage{float}
\usepackage{hyperref}
\usepackage[utf8]{inputenc}
\usepackage{listings}
\usepackage{longtable}
\usepackage{mathrsfs}
\usepackage{mathtools}
\usepackage{verbatim}
\usepackage{wasysym}
\usepackage{xcolor}
\usepackage{multirow}

\ifdefined\DEBUG
	\selectcolormodel{cmyk}
	\newcommand{\fixme}[1]{{\textcolor{red}{\bf{\textsf{FIXME: #1}}}}}
	\newcommand{\bug}[1]{{\textcolor{blue}{\bf{\textsf{BUG: #1}}}}}
	\newcommand{\idea}[1]{{\textcolor{blue}{\bf{\textsf{IDEA: #1}}}}}
	
	\newcommand{\NOTE}[1]{{\textcolor{blue}{\bf{\textsf{NOTE: #1}}}}}
	\newcommand{\TODO}[1]{{\textcolor{red}{\bf{\textsf{TODO: #1}}}}}
	\newcommand{\OLD}[1]{{\color{red}{#1}}}
	
	\newcommand{\MAR}[1]{\marginpar{{\color{red}{\textsf{#1}}}}}
\else
	\newcommand{\fixme}[1]{}
	\newcommand{\bug}[1]{}
	\newcommand{\idea}[1]{}
	
	\newcommand{\NOTE}[1]{}
	\newcommand{\TODO}[1]{}
	\newcommand{\OLD}[1]{}
	
	\newcommand{\MAR}[1]{}
\fi

\newclass{\SA}{SA}
\newclass{\NESA}{NESA}
\newclass{\CSA}{CSA}

\newclass{\NFA}{NFA}
\newclass{\tNFA}{2NFA}

\newclass{\DFA}{DFA}

\newcommand{\NN}{\mathbb{N}}
\newcommand{\NNZ}{\NN_0}

\newcommand{\setb}[2]{\{{#1}\mid{#2}\}} 

\let\lang\undefined
\DeclareMathOperator{\lang}{\operatorname{L}}
\DeclareMathOperator{\stlang}{\operatorname{S}}
\DeclareMathOperator{\langfam}{\mathcal{L}}
\DeclareMathOperator{\pref}{pref}

\newcommand{\emr}{{\,\triangleleft}}
\newcommand{\eml}{{\triangleright\,}}
\newcommand{\rw}{\rotatebox[origin=c]{180}{$\Rsh$}}
\newcommand{\purestack}[1]{\hat{#1}}

\newcommand\ipush[1]{\ensuremath{\mathtt{push(}#1\mathtt{)}}}
\newcommand\istay{\ensuremath{\mathtt{stay}}}
\newcommand\ipop{\ensuremath{\mathtt{pop}}}

\newcommand{\goesto}{\rightarrow}
\newcommand{\goestox}[1]{\mathrel{\mathop{\xrightarrow{#1}}}}

\newcommand{\movew}{\vdash_w}
\newcommand{\mover}{\vdash_r}
\newcommand{\move}{\vdash}

\newcommand{\movewp}{\vdash_w^+}
\newcommand{\moverp}{\vdash_r^+}
\newcommand{\movep}{\vdash^+}

\newcommand{\movews}{\vdash_w^*}
\newcommand{\movers}{\vdash_r^*}
\newcommand{\moves}{\vdash^*}

\newcommand{\CC}{\mathcal{C}}

\newcommand{\stacksize}[1]{{\Vert{#1}\Vert}_{\Gamma}}
\newcommand{\strlen}[1]{|{#1}|}

\newcommand{\sspace}{\sigma^{\rm s}}
\newcommand{\aspace}{\sigma^{\rm a}}
\newcommand{\wspace}{\sigma^{\rm w}}
\newcommand{\zspace}{\sigma^z}
\newcommand{\sgrow}{\mathrlap{\sigma}\raisebox{0.5pt}{\rotatebox[origin=c]{15}{$\overline{\phantom{\sigma}}$}}}

\newcommand{\alspace}{\sgrow^{\rm a}}
\newcommand{\wlspace}{\sgrow^{\rm w}}
\newcommand{\zlspace}{\sgrow^z}

\newcommand{\bigO}{O}
\newcommand{\littleO}{o}

\newcommand{\Lcomp}{L_\mathrm{comp}}
\newcommand{\Lcopy}{L_\mathrm{copy}}
\newcommand{\Lkd}{L_{\text{$k$-distinct}}}
\newcommand{\Lww}{L_{ww}}

\newclass{\Ss}{S}
\newcommand{\SCSAw}{\Ss_{\CSA}^{\mathrm{w}}}
\newcommand{\SCSAa}{\Ss_{\CSA}^{\mathrm{a}}}
\newcommand{\SCSAs}{\Ss_{\CSA}^{\mathrm{s}}}

\newcommand{\SSAw}{\Ss_{\SA}^{\mathrm{w}}}

\newcommand{\SMz}{\Ss_M^z}

\newtheorem{claim}[theorem]{Claim}
\newtheorem{conjecture}[theorem]{Conjecture}

\title{Space Complexity of Stack Automata Models
\thanks{The research of I.\ McQuillan and L.\ Prigioniero was supported, in part,
by Natural Sciences and Engineering Research Council of Canada Grant 2016-06172.}
\thanks{\textcopyright 2022. This manuscript version is made available under the CC-BY 4.0 license \url{https://creativecommons.org/licenses/by/4.0/}. Published at {\it International Journal of Foundations of Computer Science}, 32 (6), 801--823 (2021) \url{https://doi.org/10.1142/S0129054121420090}}}

\author{Oscar H. Ibarra}
\address{Department of Computer Science, University of California\\
Santa Barbara, California 93106, USA\\
\email{ibarra@cs.ucsb.edu}}

\author{Jozef Jir\'asek, Jr. and Ian McQuillan}
\address{Department of Computer Science, University of Saskatchewan\\
Saskatoon, Saskatchewan S7N 5A9, Canada\\
\email{jirasek.jozef@usask.ca, mcquillan@cs.usask.ca}}

\author{Luca Prigioniero}
\address{Dipartimento di Informatica, Università degli Studi di Milano\\
Via Celoria, 18 - Milan, Italy\\
\email{prigioniero@di.unimi.it}}

\begin{document}

\maketitle

\begin{abstract}
This paper examines several measures of space complexity of variants of stack automata:
non-erasing stack automata and checking stack automata.
These measures capture the minimum stack size required to accept every word in the language of the automaton (weak measure),
the maximum stack size used in any accepting computation on any accepted word (accept measure),
and the maximum stack size used in any computation (strong measure).
We give a detailed characterization of the accept and strong space complexity measures for checking stack automata.
Exactly one of three cases can occur:
the complexity is either bounded by a constant,
behaves 
like a linear function, or it
can not be bounded by any function of the length of the input word (and it
is decidable which case occurs).
However, this result does not hold for non-erasing stack automata;
we provide an example where the space complexity grows proportionally to the square root of the length of the input.
Furthermore, we study the complexity bounds of machines which accept a given language,
and decidability of space complexity properties.

\end{abstract}

\section{Introduction}

When studying different machine models, it is common to study both time and space complexity of a machine or an algorithm.
In particular, the study of complexity of Turing machines gave way to the area of computational complexity,
which has been an intensively studied branch of theoretical computer science for the past 40 years~\cite{Har95}.
The field of automata theory specializes in different machine models,
often with more restricted types of data stores and operations.
Various models of automata differ in the languages that can be accepted by the machine,
in the amount of information needed to describe the machine (e.g.,\ the number of states),
in the decidability of various properties of a machine,
and in the complexity of these decision problems.
As it turns out, most properties of Turing machines are undecidable,
such as whether a given machine accepts any word at all.\MAR{Reference, Rice theorem?}
Many such properties turn to be decidable when considering models with limited resources.
These include finite automata \cite{HU,harrison1978}, pushdown automata \cite{HU,harrison1978}, stack automata \cite{CheckingStack}, checking stack automata \cite{CheckingStack}, visibly pushdown automata \cite{visibly}, and many others.
As an example, the membership and emptiness problems for all the above models are decidable.

In this paper, we consider the amount of working memory that a machine uses for its computation.
For Turing machines, several different measures of space complexity have been studied.
Some of these are \cite{GiovanniTMs}:
\begin{itemize}[nosep]
\item \emph{weak measure}: for an input word $w$, the smallest tape size required for some accepting computation on $w$;
\item \emph{accept measure}: for an input word $w$, the largest tape size required for any accepting computation on $w$;
\item \emph{strong measure}: for an input word $w$, the largest tape size required for any computation on $w$.
\end{itemize}
For any of these measures, the space complexity of a machine can be defined as a function of an integer $n$
as the maximum tape size required for any input word of length $n$ under these conditions.
Finally, given a language, one can examine the space complexity of different machines accepting this language.
While these measures are well known in the context of Turing machines,%
\footnote{We point out that with regard to Turing machines,
the weak measure of complexity, corresponding to the minimal cost among all accepting computations on a given input,
is by far the most commonly used.}
they have not been extensively studied for many restricted models.
This paper aims to fill these gaps for several types of machines with a one-way read-only input tape
and a working tape that operates in stack mode: it can only be modified at one of its ends.

We study the above complexity measures for machines and languages of
stack automata, non-erasing stack automata, and checking stack automata.
Stack automata are, intuitively, pushdown automata with the additional ability to read letters from inside the stack;
but still only push to and pop from the top of the stack.
Non-erasing stack automata are stack automata without the ability to erase (pop) letters from the stack.
Finally, checking stack automata are further restricted so that as soon as they read from inside of the stack,
they can no longer modify the stack.

It is known that checking stack languages form a proper subset of non-erasing stack languages,
which form a proper subset of stack languages \cite{CheckingStack},
and those in turn form a proper subset of context-sensitive languages \cite{StackAutomataCS}.
Further results are known relating one-way and two-way versions of these machines to other models, and to space complexity classes of Turing machines, e.g.,~\cite{EngelfrietCheckingStack,DLT2018,KingWrathall}.
In terms of space complexity, it is possible to study the three space complexity measures (weak, accept, and strong)
as the maximum stack size required for any input of length~$n$.
It is already known that every stack language can be accepted by \emph{some} stack automaton
which operates in linear space using the weak measure \cite{StackAutomataCS,KingWrathall}.
However, this does not imply that \emph{every} stack automaton has this property.
We prove that every \emph{checking} stack automaton operates in linear space using the weak measure.

For checking stack automata, we give a complete characterization of the possible accept and strong space complexities.
For both measures, exactly one of the following three cases must occur for every checking stack automaton:
\begin{enumerate}[nosep]
\item The complexity is $\bigO(1)$. Then the automaton accepts a regular language.
\item There is some (accepted) word $u$ which has (accepting) computations that use arbitrarily large stack space on~$u$, and so the strong (accept) complexity is not $\bigO(f(n))$ for any integer function $f$.
The language accepted can be regular or not.
\item The complexity is $\bigO(n)$, but it is not $\littleO(n)$.
The language accepted can be regular or not.
\end{enumerate}
The third case is essentially saying that the complexity is $\Theta(n)$,
except for some technical details that will be discussed further in the paper.
Therefore, there is a ``gap'' in the possible asymptotical behaviors of space complexity.
No checking stack machine can have a space complexity between $\Theta(1)$ and $\Theta(n)$;
or complexity above $\Theta(n)$ (as long as there is \emph{some} function which bounds the space). Moreover, it is decidable to determine which one of these three scenarios occurs.
We summarize the results, alongside similar results for non-erasing and stack automata in Table \ref{tab-results1}.
The lower bound proof uses a method involving store languages of stack automata
(the language of all words occurring on the stack during an accepting computation \cite{StoreLanguages}).
We have not seen this technique used previously in literature.
Indeed, store languages are used in multiple proofs in this paper.

In the case of non-erasing stack automata, we obtain a different result:
the complexity can be $\littleO(n)$, though not constant.
We present an automaton with a weak, accept, and strong space complexities in $\Theta(\sqrt{n})$.

We also consider the following problem:
Given a language (accepted by one of the stack automaton models),
what is the lowest space complexity of any machine accepting this language, according to each of the measures?
This clearly corresponds to finding an algorithm (a Turing machine)
which solves a given problem in the least amount of working space.
For the strong measure, we show that there is a checking stack language such that
every machine accepting it needs to use arbitrarily larger stack space than the size of the input.
Therefore the strong complexity of this language is not $O(f(n))$ for any integer function $f$.
An overview on the results we have obtained on this problem is given in Table \ref{tab-results2}.

Lastly, some decidability questions on space complexity are addressed.
We show that it is undecidable whether a checking stack automaton operates in constant space using the weak measure.
However, using either the strong or accept measure,
this property turns to be decidable for any of the three studied machine models.

\MAR{Is it decidable if S/A is bounded? I think yes: detect infinite pushing lambda loops. Proven at the end, add notes.}

\section{Preliminaries}

This section introduces basic notation used in this paper, and defines the three models of stack automata that we shall consider.
The three measures of space complexity are described in the next section.

We assume that the reader is familiar with basics of formal language and automata theory.
Please see \cite{HU} for an introduction.
An {\em alphabet} is a finite set of {\em letters}.
A {\em word} over an alphabet $\Sigma = \{a_1, \ldots, a_k\}$ is a finite sequence of letters from $\Sigma$.
The set of all words over $\Sigma$ is denoted by $\Sigma^*$,
which includes the {\em empty word}, denoted by $\lambda$.
A {\em language} $L$ (over $\Sigma$) is any set of words $L \subseteq \Sigma^*$. The complement of $L$ over $\Sigma$, denoted by $\overline{L}$, is equal to $\Sigma^* \setminus L$.
Given a word $w\in \Sigma^*$, the {\em length} of $w$ is denoted by $|w|$,
and the number of occurrences of a letter $a_i$ in $w$ by $|w|_{a_i}$.
The {\em Parikh image} of $w$ is the vector $\psi(w) = (|w|_{a_1}, \ldots, |w|_{a_k})$,
which is extended to a language $L$ as $\psi(L) = \setb{\psi(w)}{w \in L}$.
\MAR{Is this used anywhere else?}
We do not define the concept of semilinearity formally here,
but it is known that a language $L$ is {\em semilinear} if and only if there is a regular language $L'$ with $\psi(L) = \psi(L')$ \cite{harrison1978}.
Given two words $u, w \in \Sigma^*$, we say that $u$ is a {\em prefix} of $w$ if $w = uv$ for some $v \in \Sigma^*$.
The {\em prefix closure} of a language $L$, $\pref(L)$, is the set of all prefixes of all words in $L$.
It is known that if $L$ is a regular language, then $\pref(L)$ is also regular.

We use the well-known big-O notation:
for two positive functions $f$ and $g$, $f(n)$ is in $\bigO(g(n))$ if there exist
positive real numbers $c$ and $n_0$ such that $f(n) \leq c g(n)$ for all $n \geq n_0$;
$f(n)$ is in $\Omega(g(n))$ if $g(n)$ is in $\bigO(f(n))$;
and $f(n)$ is in $\Theta(g(n))$ if $f(n)$ is in $\bigO(g(n))$ and $f(n)$ is in $\Omega(g(n))$.
Lastly, $f(n)$ is in $\littleO(g(n))$ if for all $c>0$, there exists $n_0$ such that for all $n > n_0$, $f(n) <c g(n)$;
and $f(n)$ is in $\omega(g(n))$ if for all $c>0$, there exists $n_0$ such that for all $n > n_0$, $f(n) > c g(n)$.

\subsection{Automata models}

Next, we define the three types of stack automata models discussed in this paper.

\begin{definition}
A one-way nondeterministic \emph{stack automaton} (\SA\ for short)
is a \mbox{6-tuple} $M = (Q, \Sigma, \Gamma, \delta, q_0, F)$, where:
\begin{itemize}[nosep]
\item $Q$ is the finite set of \emph{states}.
\item $\Sigma$ and $\Gamma$ are the \emph{input} and \emph{stack} alphabets, respectively.
\item $\Gamma$ contains symbols $\eml$ and $\emr$, which represent the \emph{bottom} and \emph{top} of the stack.
We denote by $\Gamma_0$ the alphabet $\Gamma \setminus \{ \eml, \emr \}$.
\item $q_0 \in Q$ and $F\subseteq Q$ are the \emph{initial state} and the set of \emph{final states}, respectively.
\item $\delta$ is the nondeterministic \emph{transition function}
from $Q \times (\Sigma \cup \{ \lambda \} ) \times \Gamma$
into subsets of $Q \times \setb{\istay, \ipush{x}, \ipop, -1, 0, +1}{x \in \Gamma_0}$.
We use the notation $(q,a,y) \goesto (p, \iota)$ to denote that $(p, \iota) \in \delta(q,a,y)$.
\end{itemize}
\end{definition}

A \emph{configuration} $c$ of a \SA\ is a triple $c=(q, w, \gamma)$,
where $q \in Q$ is the current state, $w \in \Sigma^*$ is the remaining input to be read,
and $\gamma$ is the current working tape (stack). The word $\gamma$ has to be either in $\eml \Gamma_0^* \rw \Gamma_0^* \emr$, or in $\eml \Gamma_0^* \emr \rw$.
The symbol~$\rw$ denotes the position of the stack head, which is currently scanning 
the symbol directly preceding it (it is not considered an element of $\Gamma$).
We shall occasionally refer to the ``pure'' stack content, that is, the word $\gamma$ without the end markers and the head symbol.
We denote this word by $\purestack{\gamma}$.
The {\em stack size} in the configuration $c$ is $\stacksize{c} = \strlen{\purestack{\gamma}} = \strlen{\gamma} -3$. 

We use two relations between configurations:
\begin{itemize}[nosep]
\item The \emph{write relation}:
If~$(q, a, x) \goesto (p, \iota)$,
where $q, p \in Q$, $a \in \Sigma \cup \{\lambda\}$, $x \in \Gamma_0 \cup \{ \eml \}$,
and $\iota \in \{ \istay, \ipush{y}, \ipop \}$ for $y \in \Gamma_0$;
then for $u \in \Sigma^*$, $\gamma \in \Gamma^*$, with $\gamma x \emr \in \eml \Gamma_0^* \emr$:
\begin{itemize}[nosep]
	\item $(q, a u, \gamma x \rw \emr) \movew (p, u, \gamma x \rw \emr)$ if $\iota = \istay$,
	\item $(q, a u, \gamma x \rw \emr) \movew (p, u, \gamma x y \rw \emr)$ if $\iota = \ipush{y}$,
	\item $(q, a u, \gamma x \rw \emr) \movew (p, u, \gamma \rw \emr)$ if $\iota = \ipop$ and $x \neq \eml$.
\end{itemize}
Notice that
the write relation is defined only if~\istay, \ipush{}, and~\ipop\ transitions are performed
when the stack head is scanning the topmost symbol of the stack.
If one of these operations is executed when the stack head is not on the top of the stack,
the next configuration is not defined.
\item The \emph{read relation}:
If~$(q, a, x) \goesto (p, \iota)$,
where $q, p \in Q$, $a \in \Sigma \cup \{\lambda\}$, $x \in \Gamma$,
and $\iota \in \{ -1,0,1 \}$;
then for $u \in \Sigma^*$, $\gamma_1, \gamma_2 \in \Gamma^*$, with $\gamma_1 x \gamma_2 \in \eml \Gamma_0^* \emr$:
\begin{itemize}[nosep]
	\item $(q, a u, \gamma_1 x \rw \gamma_2) \mover (p, u, \gamma_1 \rw x \gamma_2)$ if $\iota = -1$ and $x \neq \eml$,
	\item $(q, a u, \gamma_1 x \rw \gamma_2) \mover (p, u, \gamma_1 x \rw \gamma_2)$ if $\iota = 0$,
	\item $(q, a u, \gamma_1 x \rw \gamma_2) \mover (p, u, \gamma_1 x y \rw \gamma_2')$ if $\iota = +1$ and $\gamma_2 = y \gamma_2', y \in \Gamma$.
\end{itemize}
Note that the working tape head is not allowed to move to the left of the symbol $\eml$,
but it can move to the right of the symbol $\emr$, to detect the position at the top of the stack.
\end{itemize}
The union of $\movew$ and $\mover$ is denoted by $\move$.
The transitive closures of~$\movew$, $\mover$, and $\move$ are denoted by $\movewp$, $\moverp$, and $\movep$;
and their transitive and reflexive closures by $\movews$, $\movers$, and $\moves$, respectively.

A {\em partial computation} of the automaton $M$ on an input word $u$ is a sequence of configurations
\begin{equation}
\CC: \overbrace{(p_0,u_0, \gamma_0)}^{c_0} \move \cdots \move \overbrace{(p_n,u_n, \gamma_n)}^{c_n},
\label{configchain}
\end{equation}
where $p_0 = q_0, u_0 = u, \gamma_0 = \eml \rw \emr$.
If also $u_n = \lambda$, we say that this is a {\em computation};
and furthermore, if also $p_n \in F$ then it is an {\em accepting computation}.
The {\em stack size} of the (partial) computation $\CC$, denoted by $\stacksize{\CC}$, is defined as
$\max \setb{\stacksize{c_j}}{0 \leq j \leq n}$.

The {\em language accepted} by an automaton $M$, denoted by~$\lang(M)$,
is the set of words $w$ on which $M$ has an accepting computation.
The {\em store language} of $M$, $\stlang(M)$,
is the set of state and stack contents that can appear in some accepting computation:
$\stlang(M) = \setb{q \gamma}{(q, u, \gamma) \text{ is a configuration in some accepting computation of } M}$.
Notice that these words contain both the current state in the configuration and the stack head position marker $\rw$.
It is known that the store language of every stack automaton is a regular language \cite{KutribCIAA2016,StoreLanguages}.

The accepting computation in (\ref{configchain}) can be written uniquely as
$$c_0 \movews d_1 \moverp c_1 \movewp \cdots \movewp d_m \movers c_m.$$
We call a sequence of transitions $c_i \movews d_{i+1}$ a {\em write phase},
and a sequence of transitions $d_i \movers c_i$ a {\em read phase}.
By this definition, a computation always starts with a write phase and ends with a read phase,
even if one or both of these phases only contains a single configuration.

A stack machine is said to have \emph{partitioned states}
if the state set $Q$ is a disjoint union of two sets $Q_w$ and $Q_r$,
with $F \subseteq Q_r$,
and all transitions are
\begin{itemize}[nosep]
\item either of the form $(q_w, a, y) \goesto (p, \iota)$, with $q_w \in Q_w$, $p \in Q$, and $\iota~\in~\{\istay, \ipush{x}, \ipop \}$;
\item or of the form $(q_r, a, y) \goesto (p, \iota)$, with $q_r \in Q_r$, $p \in Q$, and $\iota~\in~\{-1, 0, 1\}$.
\end{itemize}
We denote states in $Q_w$ (resp. $Q_r$) \emph{write states} (resp. \emph{read states}),
and the transitions of the two kinds respectively \emph{write transitions} and \emph{read transitions}.
Note that the current state in a computation dictates whether the next transition to be taken is a write or a read transition.

Additionally, note that for any \SA\ $M = (Q, \Sigma, \Gamma, \delta, q_0, F)$
we can construct a \SA\ $M'~=~(Q', \Sigma, \Gamma, \delta', {q_0}_w, F')$
with partitioned states accepting the same language as $M$.
We do this by duplicating all states of $M$ and only allowing transitions which match the current state:
\begin{itemize}[nosep]
\item Let $Q' = Q_w \cup Q_r$, with $Q_w = \setb{q_w}{q \in Q}$ and $Q_r = \setb{q_r}{q \in Q}$;
\item let $F' = \setb{q_r}{q \in F}$;
\item and let $\delta'$ be a union of two transition functions:
\begin{itemize}[nosep]
\item $\delta_w$, which contains transitions $( q_w, a, y ) \goesto (p_w, \iota)$ and $( q_w, a, y ) \goesto (p_r, \iota)$;
where $(q, a, y) \goesto (p, \iota)$ in $\delta$, and $\iota \in \{\istay, \ipush{x}, \ipop \}$; and
\item $\delta_r$, which contains transitions $( q_r, a, y ) \goesto (p_w, \iota)$ and $( q_r, a, y ) \goesto (p_r, \iota)$;
where $(q, a, y) \goesto (p, \iota)$ in $\delta$, and $\iota \in \{-1, 0, 1\}$.
\end{itemize}
\end{itemize}
In order to simulate an accepting computation of $M$,
the machine $M'$ can in every step nondeterministically move to either a write or a read state
so that the next step can be performed.
Conversely, in any accepting computation of $M'$,
if we ignore the subscripts ${}_w$ and ${}_r$, we obtain an accepting computation of $M$.
Therefore $M'$ and $M$ accept the same language.

Finally, in addition to stack automata, we consider two further restricted models:
A stack automaton is called \emph{non-erasing} (\NESA)
if it contains no transitions to an element of $Q \times \{ \ipop \}$.
A non-erasing stack automaton is called a {\em checking stack automaton} (\CSA)
if it also has partitioned states and it contains no transitions from a read state to a write state.
Every accepting computation of a checking stack automaton therefore has a single write phase followed by a single read phase.

We denote by $\langfam(\SA)$, $\langfam(\NESA)$, and $\langfam(\CSA)$
the families of languages accepted by the three types of devices.

\section{Complexity Measures on Stack Automata}

We shall consider three different measures of space complexity for stack automata,
defined similarly as for Turing machines \cite{GiovanniTMs}.
Consider an input word $u \in \Sigma^*$ to a stack automaton $M$:

\begin{align*}
\intertext{The {\bf weak} measure:}
\wspace_M(u) &=
\begin{cases}
	\min \setb{\stacksize{\CC}}{\CC \text{ is an accepting computation of $M$ on $u$}}, &\text{if $u \in \lang(M)$;}\\
	0, & \text{if $u \notin \lang(M)$.}
\end{cases}
\intertext{The {\bf accept} measure:}
%
\aspace_M(u) &=
\begin{cases}
	&\mkern-40mu \max \setb{\stacksize{\CC}}{\CC \text{ is an accepting computation on $u$}},\\
	&\text{if $u \in L(M)$ and this maximum exists;}\\[7pt]
	\infty, &\parbox{10cm}{if $u \in \lang(M)$ and $M$ has accepting computations on $u$\\
		using arbitrarily large stack size;}\\[7pt]
	0, &\text{if $u \notin \lang(M)$.}\\
\end{cases}
\intertext{The {\bf strong} measure:}
\sspace_M(u) &=
\begin{cases}
	&\mkern-40mu \max \setb{\stacksize{\CC}}{\CC \text{ is a partial computation on $u$}},\\
	&\text{if this maximum exists;}\\[7pt]
	\infty, &\parbox{10cm}{if $M$ has partial computations on $u$\\
		using arbitrarily large stack size.}
\end{cases}
\end{align*}
The weak measure describes the minimum stack space required to accept the word $u$,
the accept measure describes the maximum stack space used in any accepting computation on $u$,
and the strong measure describes the maximum stack space used in any partial computation on $u$.

We extend each of the three measures to describe the space complexity of a machine
as a function of the length of the input.
Thus, for each $z \in \{{\rm w}, {\rm a}, {\rm s}\}$, we define the functions:
\begin{align*}
\zspace_M(n) &= \max \setb{\zspace_M(u)}{u \in \Sigma^* \text{ and } \strlen{u} = n}, \text{ and}\\
\zlspace_M(n) &= \max \setb{\zspace_M(u)}{u \in \Sigma^* \text{ and } \strlen{u} \leq n}.
\end{align*}
The latter essentially forces the space complexity to be a non-decreasing function.
We use the shorthand phrase ``weak complexity'' to mean ``space complexity using the weak measure'',
and similarly for the accept and strong measures.
We omit the symbol $M$ if the machine is clear from the context.

Using this notation, we can now write $\zspace(n) \in \bigO(f(n))$, $\littleO(f(n))$, $\Omega(f(n))$, etc.,
for some integer function $f(n)$ in the usual fashion.

Note that for $z \in \{{\rm a}, {\rm s}\}$, if there is even one word $u$ with $\zspace(u) =\infty$,
then $\zspace(n) = \infty$ for $n = \strlen{u}$, $\zlspace(n) = \infty$ for $n \geq \strlen{u}$,
and $\zspace(n)$ cannot be in $\bigO(f(n))$ for any integer function $f$. 
If this is the case, then we say that $M$ is {\em $z$-unlimited}, otherwise it is {\em $z$-limited}.

\noindent
\begin{example}
\label{composite}
Consider the language
$$\Lcomp = \setb{1^n}{n \text{ is a composite number}}.$$
One possible \CSA\ $M$ accepting $\Lcomp$ first non-deterministically guesses a number $d > 1$ in unary notation on the working tape.
The machine then enters read mode and moves the head back to the left end of the tape.
Next, the automaton simultaneously reads the working tape and the input tape,
rewinding back to the start of the working tape every time its end is reached, without reading any input symbols.
If the automaton reaches the end of the input simultaneously with reaching the end of the working tape,
the number of symbols $d$ on the working tape divides the number of symbols $n$ on the input.
If $M$ has traversed the working tape at least twice, then it accepts.
Observe the following facts:
\begin{itemize}
\item For any prime number $n$, $\wspace_M(n) = \aspace_M(n) = 0$,
as no string of this length is in $\lang(M)$.
\item For any composite $n$, $\wspace_M(n)$ is the smallest divisor of $n$ greater than one.
\item For any composite $n$, $\aspace_M(n)$ is the largest divisor of $n$ smaller than $n$.
\item However, $\sspace_M(n) = \infty$ for any $n$.
The automaton can guess an arbitrarily large number, and then reject when it fails to divide the input.
Therefore $M$ is ${\rm s}$-unlimited.
\end{itemize}
In terms of asymptotic behavior, first consider the accept measure.
It is clearly in~$\bigO(n)$.
However, for every prime $n$, we have $\aspace_M(n) = 0$,
so this measure is not in $\Omega(n)$.
\footnote{In fact, it is not in $\Omega(f(n))$ for any function $f(n)$ which is in $\omega(1)$.}
On the other hand, for even $n>2$, we have $\aspace_M(n) = n/2$.
Therefore this measure is not in $\littleO(n)$ either.
In this sense the measure grows ``almost linearly'', even though it is not in $\Theta(n)$.
Hence the introduction of the non-decreasing function $\alspace_M(n)$.
Its value for this machine is $n/2$ for even $n>2$, $(n-1)/2$ for odd $n>3$.
Thus this function $\alspace_M(n)$ is in $\Theta(n)$.
Finally, observe that the weak measure $\wlspace_M(n)$ is in $\bigO(\sqrt{n})$,
but it can not be bounded by a constant value.
\MAR{Rephrase}
\end{example}

\section{Space Complexities of Stack Automata}

It is known that for every stack automaton $M$
there exists another stack automaton $M'$
such that $\lang(M) = \lang(M')$ and $\wspace_{M'}(n)$ is in $\bigO(n)$ \cite{StackAutomataCS}.
For the restricted model of checking stack automata,
we can show a stronger result:
the weak space complexity of any \CSA\ is always in $\bigO(n)$.
Furthermore, if a \CSA\ is a-limited (resp. s-limited) then also the accept (resp. strong) complexity is in $\bigO(n)$.

\begin{proposition}
\label{prop-csa-linear}
Let $M = (Q, \Sigma, \Gamma, \delta, q_0, F)$ be a \CSA.
Then $\wspace_M(n)$ is in $\bigO(n)$.
\end{proposition}

\begin{proof}
First, without loss of generality, assume that the machine $M$ ends its computation
with the working head scanning the right end marker.
An automaton can be easily modified to do this without affecting its space complexity.

Let $\kappa = |Q| \cdot |\Gamma| \cdot (2^{|Q| + 1})^{4|Q|^2} + 1$.

For a contradiction assume that $M$ has a minimal (w.r.t. stack size) accepting computation $\CC$
on some word $w$ of length $n$, which uses more than $(\kappa + 1) n$ cells on the stack.
Let $s$ be the pure content of the stack in the final configuration of $\CC$.
A \emph{critical section} of $s$ is a substring of $s$,
such that the machine $M$ never reads an input symbol while the stack head is scanning one of the symbols in this section.
Since $M$ only reads $n$ input symbols, $s$ contains a critical section with length more than $\kappa$.
Denote this critical section by $u$, with $u = au'b$, $s = xau'by$, $a, b \in \Gamma$, $x, u', y \in \Gamma^*$.
The only occasions where $M$ accesses the symbols of $u$ are:
\begin{itemize}
\item
In the write phase, when it pushes $u$ on the stack:
\begin{align}
c_{i_0} = (q_{i_0}, w_0, \eml xa \rw \emr) \movews (q_{j_0}, w_0, \eml xau'b \rw \emr) = c_{j_0}.
\label{eqn-trans0}
\end{align}
\item
In the read phase, when it passes through the section using $\lambda$-transitions.
The automaton can begin at either end of the section and leave it at the same or the other end:
\begin{align}
c_{i_k} = (q_{i_k}, w_k, \eml xa \rw u'by \emr) &\movers (q_{j_k}, w_k, \eml x \rw au'by \emr) = c_{j_k} 
\label{eqn-trans1}\\
c_{i_k} = (q_{i_k}, w_k, \eml xa \rw u'by \emr) &\movers (q_{j_k}, w_k, \eml x au'b \rw y \emr) = c_{j_k}
\label{eqn-trans2}\\
c_{i_k} = (q_{i_k}, w_k, \eml xa u' \rw by \emr) &\movers (q_{j_k}, w_k, \eml x \rw au'by \emr) = c_{j_k}
\label{eqn-trans3}\\
c_{i_k} = (q_{i_k}, w_k, \eml xa u' \rw by \emr) &\movers (q_{j_k}, w_k, \eml x au'b \rw y \emr) = c_{j_k}
\label{eqn-trans4}
\end{align}
Note that the bounding configurations are chosen so that the subsequence begins after the head enters $u$,
and ends after it completely leaves the section.
\end{itemize}
By our definition of a critical section, all the transitions above are $\lambda$-transitions
and do not depend on the input word.
Our goal in the following shall be to find a replacement word $v = av'b \in \Gamma^*$, with $\strlen{v} < \strlen{u}$,
and replace each of the above subcomputations by a different sequence of valid transitions
such that the stack at the end of the entire computation will be $s' = xav'by$.
This is a contradiction with the original computation being minimal.

Since the machine $M$ does not access the input at all in any of the critical subcomputations,
we aim to simulate its functionality by an \NFA\ in the write phase
and by a two-way \NFA\ in the read phase which use the working tape as their input tape.

\begin{claim}
Let $N_0 = ((Q \times \Gamma), \Gamma, \delta_0, (q_{i_0},a), \{ (q_{j_0},b) \})$ be an \NFA\ with $\lambda$-transitions
with its transition function defined as follows:
\begin{align*}
((p, d), c) \goesto (q, c) &\text{ in $\delta_0$ if and only if } (p, \lambda, d) \goesto (q, \ipush{c}) \text{ in $\delta$, and} \\
((p, d), \lambda) \goesto (q, d) &\text{ in $\delta_0$ if and only if } (p, \lambda, d) \goesto (q, \istay) \text{ in $\delta$.}
\end{align*}

Then a word $v = av'b$ is accepted by $N_0$
if and only if $M$ has some valid sequence of transitions
$$(q_{i_0}, w_0, \eml xa \rw \emr) \movews (q_{j_0}, w_0, \eml xav'b \rw \emr).$$
\end{claim}
Note: The machine $N_0$ needs to remember the last symbol it has read in its state control.
This corresponds to the last symbol pushed on the top of the stack of $M$.

For the next claim, we use two-way $\NFA$s which have input surrounded by end-markers $\eml$ and $\emr$, and 
transitions that move left ($\mathrm{L})$, right ($\mathrm{R})$, or stay ($\mathrm{S})$.
\begin{claim}
Let $N_k = (Q \cup \{ q_a \}, \Gamma, \delta_k, q_{i_k}, \{ q_a \})$ be a two-way \NFA\
with a transition function defined as follows:
\begin{align*}
(p, c) \goesto (q, \mathrm{L}) &\text{ in $\delta_k$ if and only if } (p, \lambda, c) \goesto (q, -1) \text{ in $\delta$,} \\
(p, c) \goesto (q, \mathrm{S}) &\text{ in $\delta_k$ if and only if } (p, \lambda, c) \goesto (q, 0) \text{ in $\delta$,} \\
(p, c) \goesto (q, \mathrm{R}) &\text{ in $\delta_k$ if and only if } (p, \lambda, c) \goesto (q,+1) \text{ in $\delta$, and} \\
(q_{j_k}, \eml) \goesto (q_a, \mathrm{S})
\end{align*}

Then a word $v = av'b$ is accepted by $N_k$
if and only if $M$ has some valid sequence of transitions
$$(q_{i_k}, w_k, \eml xa \rw v'by \emr) \movers (q_{j_k}, w_k, \eml x \rw av'by \emr).$$
\end{claim}
Note: This automaton deals with case (\ref{eqn-trans1}).
Similar automata for the other three cases (\ref{eqn-trans2}) - (\ref{eqn-trans4}) can be obtained
by checking for the right end marker in the accepting state;
or by first moving the head to read the last symbol of the input word.

Continuing with the proof of Proposition \ref{prop-csa-linear}:
Observe that the construction of the automata $N_0$ and $N_k$
only depends on the machine $M$, not on the input word.
Further, there can be at most $4 |Q|^2$ distinct automata $N_k$,
since their construction differs only in the states $q_{i_k}$ and $q_{j_k}$,
and in the direction the machine enters and leaves the critical section in.
The constructed \NFA\ $N_0$ and the two-way \NFA s $N_k$ accept regular languages.
Since there are only $4 |Q|^2 + 1$ such languages under consideration,
the intersection of all these languages must be regular as well.
Denote this language by $V$.

Let us examine the number of states required by a one-way \NFA\ accepting $V$.
The language $N_0$ can be accepted by a $(|Q| \cdot |\Gamma|)$-state \NFA.
Each of the languages $N_k$ can be accepted by a $(|Q| + 1)$-state two-way \NFA,
for which we can construct an equivalent one-way \NFA\ with $\binom{2(|Q| + 1)}{|Q| + 2}$ states; see \cite{Kapoutsis}.
For our argument it is enough that this is less than $2^{|Q| + 1}$.
Finally, the state complexity of an intersection of several \NFA s is bounded by the product of their sizes.
In our case, we can construct a one-way \NFA\ accepting $V$
using at most $|Q| \cdot |\Gamma| \cdot (2^{|Q| + 1})^{4|Q|^2} < \kappa$ states.

Now recall the discussion of the critical section $u$.
As we have $u \in V$, $V$ can be accepted by an \NFA\ with at most $\kappa$ states, and $\strlen{u} > \kappa$,
we can use the pumping lemma for regular languages to partition the string $u$
into three substrings $u = u_1 u_2 u_3$,
such that $\strlen{u_1 u_3} < \kappa$, $\strlen{u_2} > 0$, and $v = u_1 u_3$ is in $V$.

As $v$ is accepted by $N_0$ and all machines $N_k$,
by the above results we can replace each of the sequences of transitions (\ref{eqn-trans0}) - (\ref{eqn-trans4})
by the transitions given as a result of one of the above lemmas.
And since the automaton never accesses the critical section of the stack other than in those five cases,
the resulting computation is still a valid accepting computation of $M$.
However, this computation only uses a stack of size $\strlen{ xvy } < \strlen{x} + \kappa + \strlen {y} \leq \strlen{ xuy }$,
which is a contradiction with our assumption that the original computation was minimal with respect to stack size.
\end{proof}

We obtain a similar result for the accept and strong measures.
However, consider a checking stack automaton accepting $\Sigma^*$
which first non-deterministically pushes an arbitrary word on the stack, then reads the entire input, and accepts.
This machine has an unlimited accept and strong complexity.
Thus the result of Proposition \ref{prop-csa-linear} cannot always hold for these two measures.
But we shall prove that as long as this complexity is limited, it still has to be at most linear.
Further, we shall see in Proposition \ref{decidablelimited} that it is decidable whether a \CSA\ is a- or s-limited or not.

\begin{proposition}
\label{acceptupperbound}
Let $M = (Q, \Sigma, \Gamma, \delta, q_0, F)$ be a $\CSA$.
For $z \in \{{\rm a}, {\rm s} \}$,
if $M$ is $z$-limited, then $\zspace_M(n)$ is in $\bigO(n)$.
\end{proposition}
\begin{proof}
Let $\kappa$ be the constant from the proof of Proposition \ref{prop-csa-linear}.
For a contradiction, assume that there exists some string $w$ of length $n$
and some accepting computation (for the accept measure) or some computation (for the strong measure)
of $M$ on $w$ which uses more than $(\kappa + 1)n$ tape cells.
Similarly to the proof of the previous result,
we can find a critical section $u$ of the final stack contents $s$ in this computation,
such that the machine never reads any input while the stack head is scanning a symbol in this section,
and $\strlen{u} > \kappa$.

Also, like in the previous proof, using the pumping lemma we can write $u$ as $u_1 u_2 u_3$,
such that now $\strlen{u_1 u_3} < \kappa$, $\strlen{u_2} > 0$, and $v = u_1 u_2^i u_3 \in V$ for any $i \in \NNZ$.
However, following the previous proof, we can again obtain an accepting computation (resp. computation)
of $M$ on $w$ in which the critical section $u$ of the stack is replaced by $v$.
And since we can choose the string $v$ to be of arbitrary length,
this is a contradiction with the assumption that $M$ is $z$-limited.

Therefore every accepting computation (resp. every computation) on every string of length $n$
uses at most $(\kappa + 1)n$ working tape cells,
and thus $\zspace_M(n)$ is in $\bigO(n)$.
\end{proof}

Whether or not the result above holds for $\NESA$ and $\SA$ generally is an open problem.

Lower bounds on the space complexity functions can also be studied similarly to upper bounds. The next proof starts with an accepting computation using some stack word, and then finds a new accepting computation on some possibly different input word that is roughly linear in the size of the stack. It then uses the regularity of the store languages of stack automata in order to determine that for every increase in some constant, there is at least one more final stack word of that length that has an input word whose size is linear in the size of that stack contents. That is enough to show that the accept and strong space complexities cannot be $o(n)$, and if the non-decreasing function $\zlspace$ is used, then it is at least linear.

\MAR{Note: I think I can prove this also for w measure using a different method.
Please skip this proof for now.}

\begin{lemma}
\label{lem-csa-atleastlin}
Let $z \in \{{\rm a}, {\rm s}\}$.
Let $M = (Q,\Sigma,\Gamma, \delta,q_0,F)$ be a $\CSA$ such that $\zspace(n) \notin O(1)$ and $M$ is $z$-limited. Then
$\zspace(n)$ cannot be $o(n)$, and
$\zlspace(n) \in \Omega(n)$.
\end{lemma}
\begin{proof}
First, consider the case $z = {\rm a}$.

We shall show that there are constants $c',d'$ such that, for all final stack contents $\gamma$ that can appear in an accepting computation on $u = u_1 u_2$, where $u_1$ is the input read during the write phase, and $u_2$ is the input read during the
read phase, there are words $\bar{u}_1$ and $\bar{u}_2$ that can be derived from~$u_1$ and $u_2$ whereby $\bar{u}_1 \bar{u}_2$ is also accepted with an accepting computation with final stack $\gamma$ where $\bar{u}_1$ is the input read during the write phase, $\bar{u}_2$ is the input read during the read phase, and $\frac{|\bar{u}_1| +|\bar{u}_2|}{d'} \leq |\gamma|$
if $|\gamma|>0$ and $|\bar{u}_1|+ |\bar{u}_2| \leq c'$ if $|\gamma| = 0$. To find $\bar{u}_1$, do the following. Let $m = |Q|$. For every part of the derivation of $u_1$ where more than $m$ consecutive transitions are applied between writing letters to the stack, we can ``cut out'' those parts of $u_1$ between the same state, and make a smaller word $\bar{u}_1$ that satisfies $|\bar{u}_1| \leq m (|\gamma|+1)$.  We can similarly ``shrink'' $u_2$ into another word $\bar{u}_2$ as follows: if, while reading $u_2$, we hit the same position of $\gamma$ in the same state twice, we can cut out the parts of $u_2$ between those, which still gives an accepting computation. Then $|\bar{u}_2| \leq m (|\gamma|+1)$. Thus, the modified input $\bar{u}=\bar{u}_1 \bar{u}_2$ accepts while writing
$\gamma$, and when reading it, is of length at most $ 2 m (|\gamma|+1)$, and hence
$\frac{|\bar{u}|}{4m} \leq |\gamma|$ if $|\gamma| >0$ and $|\bar{u}| \leq 2m$ otherwise.

Furthermore, by Proposition \ref{acceptupperbound}, there are numbers $e$ and $n_0'$ such that $M$ must use a stack of at most size $e \cdot n$ on an input of size $n \geq n_0'$. For $\bar{u}$ with $|\bar{u}| \geq n_0'$, this implies $|\gamma| \leq e |\bar{u}|$. Thus, there exist $d,e \geq 0$ such that $d|\bar{u}| \leq |\gamma| \leq e|\bar{u}| $.
Hence,
we know that given any final stack word $\gamma$ that can appear in an accepting computation, there is some input word $u$ that gives $\gamma$ as the final stack such that $d|u| \leq |\gamma| \leq e|u|$.

In addition, we know that the store language of $M$, $S(M)$, is regular because the store language of every stack automaton (and hence every checking stack automaton) is a regular language \cite{StoreLanguages}. 
Furthermore, it is possible to restrict $S(M)$ to only stack contents that can appear at the end of write phases and
not partial contents of the stack, by intersecting the regular $S(M)$ with words that start with some final state
of $M$ (all words in $S(M)$ start with the state), and then to erase the read head symbol and the state symbol to
obtain $S'(M) \subseteq \Gamma^*$ of final stack contents that can appear in accepting computations. Therefore, $S'(M)$ is 
regular as well.

If $S'(M)$ is finite, then $\aspace(n)$ is $O(1)$, thus yielding a contradiction.
Thus, $S'(M)$ must be infinite, and by the pumping lemma there are strings $x$, $y$, $z$ with $\strlen{y} > 0$
such that $\gamma_i = x y^i z$ is in $S'(M)$ for all $i \in \NNZ$.
Let $c = \strlen{xz}$ and $p = \strlen{y}$.

Thus, there exists an infinite sequence $\gamma_0, \gamma_1, \ldots$ of final stack contents appearing in accepting computations with $|\gamma_i| = c + i p$ for all $i \geq 0$, 
and there exists constants $d,e$ such that, for each $\gamma_i$, there exists
$v_i \in L(M)$ which has an accepting computation with final stack contents $\gamma_i$, and
$d|v_i| \leq |\gamma_i| \leq e |v_i|$.
Thus, $\frac{|\gamma_i|}{e} \leq |v_i| \leq \frac{|\gamma_i|}{d}$. Hence, for each $i$, since
$\alspace(n)$ is non-decreasing, $\alspace( \lfloor \frac{|\gamma_i|}{d} \rfloor) \geq |\gamma_i|$;
that is, $\alspace( \lfloor \frac{c + i p}{d} \rfloor) \geq c + i p$, and $\alspace(j)$ does not decrease for any
$j > \lfloor \frac{c + i p}{d} \rfloor$.

Let $n_0 = \lfloor \frac{c}{d}\rfloor$.
We will prove next that for all $n > n_0$, $\alspace(n) \geq  \frac{d}{p+1} n$ which will prove
$\alspace(n) \in \Omega(n)$. Let $n > \frac{c}{d}$.
Then, $n$ is between
$\lfloor \frac{c + i p}{d} \rfloor$ and $\lfloor \frac{c + (i+1) p}{d} \rfloor -1$, for some $i$, and
$\alspace(n) \geq c + i p$. Letting $r = c + i p$, $n$ is between
$\lfloor \frac{r}{d} \rfloor$ and $\lfloor \frac{r + p}{d} \rfloor -1$ and $\alspace(n) \geq r$.
Thus, for any $n$ between $\lfloor \frac{r}{d} \rfloor$ and $ \frac{r(p+1)}{d}$, also $\alspace(n) \geq r$ since
$\alspace(n)$ is non-decreasing. Therefore, $\alspace(n) \geq n \frac{d}{p+1}$ for all $n>n_0$. Hence, $\alspace(n) \in \Omega(n)$.

A similar result follows from this for the strong measure $\sspace(n)$ by making a new machine from $M$ with all
states final and using the same analysis (it suffices to show it is at least linear on just computations rather than
all partial computations and making all states final will accomplish this).
\end{proof}

Lemma \ref{lem-csa-atleastlin} is the ``best possible result'' in that it is not always the case that~$\aspace(n) \in \Omega(n)$ since the space complexity can periodically go below linear infinitely often as demonstrated with Example \ref{composite}. But what this lemma says is that it returns to at least linear infinitely often as well. Furthermore, if one forces the complexity to be non-decreasing, then it is at least linear.
Example \ref{composite} also shows that the result of Lemma \ref{lem-csa-atleastlin} does not hold for the weak measure.

\begin{proposition}
\label{sqrt}
There exists a \CSA\ $M$ which accepts a non-regular language
such that $\wspace_M(n)$ is in $\bigO(\sqrt{n})$, but not in $\littleO(\sqrt{n})$.
\end{proposition}

\begin{proof}
Consider the automaton $M$ described in Example \ref{composite} accepting the language $\Lcomp$.
For any prime number $n$, we have $\wspace_M(n) = 0$.
For a composite $n$, $\wspace_M(n)$ is the smallest divisor of $n$ greater than one,
which is at most $\sqrt{n}$.
This means that $\wspace_M(n)$ is in $\bigO(\sqrt{n})$.
For a number $n$ which is the square of a prime, we have $\wspace_M(n) = \sqrt{n}$.
Therefore there are infinitely many such numbers where this holds,
and thus $\wspace_M(n)$ is not in $\littleO(\sqrt{n})$.
\end{proof}

Putting together all results for $\CSA$ with either the accept or strong measure so far, we get the following complete characterization:
\begin{theorem}
\label{thm-csa-full}
Let $M$ be a $\CSA$. For $z \in \{{\rm a}, {\rm s}\}$, exactly one of the following must occur:
\begin{enumerate}[nosep]
\item $M$ is $z$-unlimited, and so there is no $f$ such that $\zspace(n) \in O(f(n))$.
(In other words, $\zspace$ takes on an infinite value for some input.)
In this case, $L(M)$ can be either regular or not.
\item $M$ is $z$-limited, $\zspace(n) \in O(1)$, and $L(M)$ is regular;
\item $M$ is $z$-limited, $\zspace(n) \in O(n), \zspace \notin o(n)$, and $\zlspace(n) \in \Theta(n)$ (and $L(M)$ can be either regular or not).
\end{enumerate}
\end{theorem}
\begin{proof}
\MAR{
Is the ``can be either regular or non-regular'' distinction here actually useful?
Except for the seond case?
}
Any \CSA\ can be modified so that it first pushes an arbitrary number of a new symbol $x$ on the stack,
and then resumes the original computation, treating the symbol $x$ as the left end marker.
This modified \CSA\ accepts the same language, however it is not accept and strong limited.
Since \CSA s can accept both regular and non-regular languages,
either of the possibilities in case (1) can happen.

If $\zspace_M(n)$ is in $O(1)$, there is some constant $c$ such that any (accepting) computation of $M$
requires only $c$ work tape cells.
Therefore it is possible to simulate the work tape using
only a non-deterministic state control unit with $|Q| \cdot |\Gamma|^c$ states,
and thus $\lang(M)$ is regular.

If $\zspace_M(n)$ is $z$-limited, but not in $O(1)$, Lemma \ref{lem-csa-atleastlin} applies,
and we obtain the third case.
For a regular example, consider a \CSA\ which copies its input on the work tape and accepts.
For a non-regular example, consider a \CSA\ accepting the language $\setb{w \# w}{w \in \Sigma^*}$
by copying the first part of the input on the stack, entering read mode,
and comparing the rest of the input with the contents of the stack.
\end{proof}

We conclude this section by giving an example of a \NESA\ for which Theorem \ref{thm-csa-full} does not hold.
Therefore the requirement for the machine to be a checking automaton is necessary. 

\begin{proposition}
\label{lessthann}
There exists a $\NESA$ (and a $\SA$) $M$ which accepts a non-regular language such
that $\wspace(n)$, $\aspace(n)$, and $\sspace(n)$ are all in $\Theta(\sqrt{n})$.
\end{proposition}
\begin{proof}
Consider the language
$L = \setb{a^1 b a^2 b \cdots a^r b}{r \geq 0 }$,
and let $L_0 = \pref(L)$,
which is not regular.

The language $L_0$ can be accepted by a \NESA\ $M$ in the following way:
\begin{enumerate}
\item Push a symbol $a$ on the stack.
\item Move the head to the left end of the work tape.
\item While both the input and the work head scan symbols $a$, move them to the right synchronously.
\item If the end of the input is reached, accept.
\item If a symbol $b$ is read from the input tape and the head is at the right end of the work tape,
repeat from step (1). Otherwise reject the input.
\end{enumerate}

For every $n \in \NNZ$, there is exactly one word of length $n$ in $L_0$.
The length of the stack $s$ used in the only computation of $M$ on this word satisfies
$\sqrt{n/2} - 2 \leq \strlen{s} \leq 2 \sqrt{n} + 1$.
Further, no input string (whether or not in $\lang(M)$) uses more than $2\sqrt{n}+1$ work tape cells.
Thus all three complexity measures are in $\Theta(\sqrt{n})$.
\end{proof}

In conclusion, for \CSA, the accept and strong measures have been completely characterized, as either constant, linear, or $z$-unlimited. For the weak measure, each machine has $\wspace(n)$ in $\bigO(n)$. If $\wspace(n)$ is in $\bigO(1)$, then $\lang(M)$ is regular; otherwise there are machines with $\wspace(n) \in \Theta(n)$, but we also demonstrate a machine $M$ with $\wspace(n)$ in $\bigO(\sqrt{n})$ but not $\littleO(\sqrt{n})$. The other possible functions are unknown. For $\NESA$ and $\SA$, unlike in the case of \CSA, non-erasing stack automata with an accept or strong space complexity
of $\Theta(\sqrt{n})$ exist for all three measures.
The question of which functions can asymptotically describe this behavior, other than $\Theta(\sqrt{n})$, remains open.
These results are summarized in Table \ref{tab-results1}.
\begin{table}[th!]
\centering
\begin{tabular}{|l||ccc||ccc||}
\hline
                                           & \multicolumn{3}{c||}{$\CSA$}  & \multicolumn{3}{c||}{$\NESA$ and $\SA$}          \\ \hline
function                                   & weak         & accept       & strong       & weak         & accept       & strong \\\hline
$\bigO(1)$                                 & $\checkmark$ & $\checkmark$ & $\checkmark$ & $\checkmark$ & $\checkmark$ & $\checkmark$ \\
between                                    & ?            &              &              & ?            & ?            & ?       \\
$\bigO(\sqrt{n})$ not $\littleO(\sqrt{n})$ & $\checkmark$ &              &              & $\checkmark$ & $\checkmark$ & $\checkmark$ \\
between                                    & ?            &              &              & ?            & ?            & ?       \\
$\bigO(n)$ not $\littleO(n)$               & $\checkmark$ & $\checkmark$ & $\checkmark$ & $\checkmark$ & $\checkmark$ & $\checkmark$ \\
between                                    &              &              &              & ?            & ?            & ?           \\
unlimited                                  &              & $\checkmark$ & $\checkmark$ &              & $\checkmark$ & $\checkmark$ \\\hline
 \end{tabular}
\caption{For each machine model with $\CSA$ in columns 2 through 4, and both $\NESA$ and $\SA$ in columns 5 through 7 (the results are the same for both), for each of the space complexity measures, weak, accept, or strong, and for each possible function represented in the rows, a $\checkmark$ is placed if there exists a machine of that type with that measure that gives exactly that function, a blank space is left if there does not exist such a machine with that function, and a ? is given if it is open whether there exists such a machine.}
\label{tab-results1}
\end{table}

\section{Space Complexities of Languages Accepted by Stack Automata}

So far, we have been concerned with measuring the space complexity of a specific stack machine.
A natural question to ask is, given a language,
what is the minimum space complexity of a (non-erasing, checking) stack machine which accepts this language.
We define this as follows:

\begin{definition}
Let $L$ be a language and $f(n)$ be a function of natural numbers.

We say that $L$ belongs to the class $\SCSAw(f(n))$,
if there exists a \CSA\ accepting $L$ with weak space complexity in $\bigO(f(n))$,
and there is no \CSA\ accepting $L$ with weak space complexity in $\littleO(f(n))$.

The language classes $\SCSAa(f(n))$ and $\SCSAs(f(n))$ are defined analogously for the accept and strong complexity measures of automata.


We say that $L$ belongs to the class $\SCSAa(\infty)$ (resp. $\SCSAs(\infty)$),
if every \CSA\ accepting $L$ is ${\rm a}$-unlimited (resp. ${\rm s}$-unlimited).

Similar classes of \NESA\ and \SA\ languages are defined analogously.
\end{definition}

As an example, consider the language $\Lww = \setb{w \# w}{w \in \{ a, b \}^*}$.
The naive \CSA\ accepting $\Lww$ by copying the first instance of the string $w$ to the work tape
and comparing it to the second instance has an accept space complexity in $\bigO(n)$.
Since $\Lww$ is not regular, according to Theorem \ref{thm-csa-full},
any \CSA\ accepting $\Lww$ is either a-unlimited,
or its accept space complexity is not in $\littleO(n)$.
Therefore we can say that $\Lww$ is in the class $\SCSAa(n)$.

Proposition \ref{prop-csa-linear} shows that for any $L \in \langfam(\CSA)$,
there is a checking stack automaton $M$ which accepts $L$ with $\wspace_M(n)$ in $O(n)$.
A similar result also holds for languages of all stack automata \cite{StackAutomataCS,KingWrathall}.
Therefore there are no languages belonging to the classes $\SCSAw(f(n))$ and $\SSAw(f(n))$
for any function $f(n)$ which is in $\omega(n)$.
Informally speaking, the weak space complexity of any (checking) stack automaton language is at most linear.

Next we shall consider the accept and strong classes of languages.
We show that there exist \CSA\ languages in $\SCSAs(\infty)$,
i.e. languages for which there is no accepting \CSA\ that is s-limited.

\begin{example}
\label{examplenotstrong}
Consider the language
$$\Lcopy=\setb{ u \$ u \# v \$ v }{ u,v \in \{a,b\}^* }.$$
There is a \CSA\ that accepts $\Lcopy$ as follows:
the automaton first guesses two strings $u'$ and $v'$ on $\lambda$-transitions,
pushes $u' \# v'$ on the stack,
and finally it verifies that the guessed strings match both occurrences of strings $u$ and $v$ on the input tape.

Using the accept measure, it is easy to see that $\aspace(n)$ for this machine is in $\bigO(n)$:
the only accepting computation is the one in which the machine guesses $u$ and $v$ correctly.
However, if we consider the strong measure,
the automaton can guess arbitrarily long strings $u'$ and $v'$, and then reject.
Even if the computation is not accepting, it is still counted towards the strong measure.
Therefore this machine is ${\rm s}$-unlimited.
\end{example}

We can ask whether \emph{every} machine that accepts $\Lcopy$ is s-unlimited.
We will see that this is in fact the case.
To prove this, we shall use the following lemma. 

\begin{lemma}
\label{lem-lwm}
Let $M = (Q, \Sigma, \Gamma, \delta, q_0, F)$ be a \CSA.
Consider the language consisting of prefixes of input words which $M$ reads during the entire write phase in some accepting computation:
$$L_{w,M} = \setb{u}{
(q_0, uv, \eml \rw \emr) \vdash_w^*
(q, v, \gamma) \vdash_r^*
(q_f, \lambda, \gamma'),
q \in Q_r,
q_f\in F
}.$$
The language $L_{w,M}$ is a regular language.
\end{lemma}
\begin{proof}
We shall construct an \NFA\ $N$ accepting this language.
The basic idea of the construction is to simulate two machines in parallel:
first, the machine $M$ during its write phase,
and second, an \NFA\ accepting the store language of $M$.
Denote the latter by $S = (Q_S, Q \cup \Gamma \cup \{\rw\}, \delta_S, s_0, F_S)$.
Since in its write phase the machine $M$ can not examine the contents of its stack beyond the top symbol,
it can be simulated by an \NFA.
Symbols pushed on the stack are then read by the \NFA\ $S$,
which accepts at the end of the simulated write phase if the content of the stack belongs to the store language,
and thus it occurs during some accepting computation of $M$.

The state set of $N$ is $(Q \times Q \times \Gamma \times Q_S)$, and an additional initial state $n_0$.
Recall that the store language of $M$ contains strings of the form $q \gamma$,
where $q$ is a state of $M$ and $\gamma$ is the content of the stack including symbols $\eml$, $\emr$, and $\rw$.
Since the first symbol of $q \gamma$ is the state, it is read and stored
in the first component of the states of $N$ (which does not change throughout the computation).
The second component is a state of $M$ that simulates the write phase of $M$,
with the third component being the simulated top of the stack,
and the fourth component is the state of the simulated store language automaton.


The detailed construction of this \NFA\ $N = (Q_N, \Sigma, \delta_N, n_0, F_N)$ is as follows:
\begin{align*}
Q_N &= ( Q \times Q \times \Gamma \times Q_S ) \cup \{ n_0 \}
\end{align*}
\begin{align*}
n_0 &\goestox{ \eml} (p, q_0, \eml, s)
&&\text{if $s_0 \goestox{p \eml} s$ in $\delta_S$, where $p \in Q$,} \\
(p, q_1, x, s) &\goestox{a} (p, q_2, x, s)
&&\text{if $(q_1, a, x) \goesto (q_2, \istay)$ in $\delta$,} \\
& &&\text{where } p, q_1, q_2 \in Q; s \in Q_S; x \in \Gamma; a \in \Sigma \cup \{ \lambda \}, \\
(p, q_1, x, s_1) &\goestox{a} (p, q_2, y, s_2)
&&\text{if $(q_1, a, x) \goesto (q_2, \ipush{y})$ in $\delta$ and $s_1 \goestox{y} s_2$ in $\delta_S$,} \\
& &&\text{where } p, q_1, q_2 \in Q; s_1, s_2 \in Q_S; x, y \in \Gamma; a \in \Sigma \cup \{ \lambda \}.
\end{align*}
%
The final states of $N$ are states of the form $(p, p, x, s)$,
where $p$ is a read state of $M$,
$x \in \Gamma$,
$s \in Q_S$, and
$s \goestox{\rw \emr} s_f$ in $\delta_S$ for some final state $s_f$ of $S$.

This machine $N$ accepts a string $u$ in a state $(p, p, x, s)$
if $(q_0, u, \eml \rw \emr) \movews (p, \lambda, \gamma)$,
and $p \gamma \in \stlang(M)$,
and $p$ is a read state of $M$.
If the string $p \gamma$ is in the store language of $M$,
there must be some computation of $M$ of the form
$$
(q_0, u'v, \eml \rw \emr) \moves (p, v, \gamma) \movers (q_f, \lambda, \gamma') \qquad \text{with $q_f \in F$.}
$$
The second part of the computation must consist of only read transitions,
since $p$ is a read state and $M$ is a \CSA.
Then chaining the computations together we get
$$
(q_0, uv, \eml \rw \emr) \movews (p, v, \gamma) \movers (q_f, \lambda, \gamma')
$$
and therefore $u \in L_{w, M}$.
\end{proof}

We can now prove the following result:
\begin{proposition}
\label{candidatestrong}
Every \CSA\ $M$ accepting the language $\Lcopy$ is ${\rm s}$-unlimited.
\end{proposition}
\begin{proof}
Let $M = (Q, \{ a, b, \$, \# \}, \Gamma, \delta, q_0, F)$
be a \CSA\ accepting $\Lcopy$.
Consider the following subset of $\Lcopy$:
$$L' = \setb{ u \$ u \# v \$ v }{ u, v \in \{ a, b \}^*; \strlen{v} > ((2 \strlen{u} + 1) |Q| |\Gamma| + 1) |Q| }$$
First assume that every accepting computation of $M$ on every string in $L'$ enters the first read state
only after reading the symbol $\#$.
By Lemma \ref{lem-lwm}, the language $L_{w,M}$ is regular,
and thus so is $\pref(L_{w,M}) \cap ( \{ a, b \}^* \$ \{ a, b \}^* \# )$.
However, the latter language is by our assumption equal to $\setb{u \$ u \#}{u \in \{a, b\}^*}$, which is not regular;
this is a contradiction.

Therefore let $w = u \$ u \# v \$ v$ be a string in $L'$ on which $M$ has an accepting computation
which enters the first read state before $M$ reads the symbol $\#$.
Since $M$ is a \CSA, the contents of the work tape are fixed by the time $M$ reads the substrings $v$.
Let $\kappa = (2 \strlen{u} + 1) |Q| |\Gamma|$,
and compare the total number of work tape cells used by $M$ in this computation to $\kappa$.

Assume that the total number of work tape cells used by $M$ in this computation is at most $\kappa$.
Thus the machine $M$ has in its read phase at most $(\kappa + 1) |Q|$ distinct configurations,
considering the position of the tape head and the current state.
By our choice of $L'$, at least two configurations reached while reading symbols from the second occurrence of the string $v$ must be exactly identical.
This however means that we can omit the sequence of configurations between these two,
and obtain an accepting computation of $M$ on the string $u \$ u \# v \$ v'$, where $\strlen{v'} < \strlen{v}$.
This is a contradiction with the fact that $\lang(M) = \Lcopy$.

Finally, assume that this computation on the string $w$ uses more than $\kappa$ tape cells.
All of these symbols are pushed on the stack before $M$ reads the symbol $\#$.
This means that there must be some sequence of at least $|Q| |\Gamma|$ subsequent configurations in this computation
during which $M$ does not read any input, but pushes at least one symbol on the stack.
This computation therefore contains a loop in which $M$ returns to the same state with the same symbol on the top of the stack using only $\lambda$-transitions, and pushes at least one symbol on the stack.
By repeating this loop we can obtain a partial (not necessarily accepting) computation of $M$ on the same input $w$
which writes an arbitrary number of symbols on the work tape.
Therefore, $M$ is s-unlimited.
\end{proof}

The existence of \CSA\ languages in $\SCSAa(\infty)$ is still open, as discussed next.
\begin{conjecture}
Consider the language 
$$\Lkd =
\{ 1^k \# v_1 \# \cdots \# v_m  \mid
 v_i \in \{ a, b\}^*, | \{v_i \mid 1 \leq i \leq m\} | \leq k \}.$$
A string in this language begins by giving a number $k$ in unary notation,
and then some number of strings over $\{ a, b \}$ separated by markers.
A string belongs to the language if among the $m$ strings $v_i$ there appear at most $k$ distinct ones.
For example, the string $1 1 \# a \# ab \# a$ is in the language, while $1 1 \# a \# ab \# ba$ is not.

We conjecture that every checking stack automaton accepting $\Lkd$ is ${\rm a}$-unlimited.
\end{conjecture}

The language $\Lkd$ can be accepted by a \CSA\ which,
for every $1$ read, non-deterministically guesses a word $u_i$ on the stack, separated by markers $\$$.
After reading the first symbol $\#$, the work tape therefore contains $\$ u_1 \$ \cdots \$ u_k \$$.
The machine switches to read mode.
Upon reading each $\#$,
the machine now non-deterministically moves to one of the markers $\$$ on the stack,
and compares the input string $v_i$ to the string $u_j$ on the work tape.
If they match, the computation continues; if they disagree, the machine halts and rejects.

If the input contains at most $k$ distinct strings $v_i$,
it is indeed possible for the machine to guess these strings
(guessing the same string multiple times if there are fewer than $k$ distinct ones),
and then for every input string $v_i$ nondeterministically find the corresponding string $u_j$ on the stack and match them.
But if the input contains more than $k$ distinct substrings $v_i$,
then some of them must remain unmatched to the $k$ guessed strings,
and the machine rejects in every computation.

This machine is, however, a-unlimited.
Consider the input $11 \# a \# a$.
This string is in $\Lkd$.
Since the automaton is allowed to guess two strings,
as long as one of the guessed strings is $a$,
the other can be arbitrarily long and the input can still be matched.
Therefore accepting computations on this string exist using arbitrarily large amount of space.
We conjecture that this holds for any $\CSA$ that recognizes the language $\Lkd$.

Results from this section are summarized in Table \ref{tab-results2},
which shows the known non-empty classes of languages.
For all of $\CSA, \NESA, \SA$, those languages that can be accepted with a machine with space complexity in $O(1)$ are exactly the regular languages.
By Theorem~\ref{thm-csa-full},
every \CSA\ machine has an accept space complexity of either
$O(1)$, $\bigO(n)$ but not $\littleO(n)$, or is a-unlimited.
Therefore every \CSA\ language belongs to either $\SCSAa(1)$, $\SCSAa(n)$, or $\SCSAa(\infty)$.
The same holds for the strong measure.
%
We have proven in Proposition \ref{candidatestrong} that \CSA\ languages in $\SCSAs(\infty)$ exist;
this question is still open for the class $\SCSAa(\infty)$.
For the weak measure, the only known result for $\SA$ is that for every stack automata language, there is some machine that accepts it
in at most linear space \cite{StackAutomataCS}, and so there is no language that requires more than linear space; this is also true for $\CSA$ by Proposition \ref{prop-csa-linear}. This is open for non-erasing stack automata as the construction in \cite{StackAutomataCS} introduces transitions that pop.

Finally, it is an interesting open question to ask which are all the possible functions,
aside from constant and linear, that define a non-empty class of \CSA\ (\NESA, \SA) languages.
In Propositions \ref{sqrt} and \ref{lessthann} we have seen automata which operate with a space complexity in $\bigO(\sqrt{n})$.
However, the exact classes which these languages belong to are still unknown,
as we have not proven that there is no machine accepting that language with a space complexity in $\littleO(\sqrt{n})$.

\begin{table}[th!]
\centering
\begin{tabular}{|l||ccc||ccc||}
\hline
                                           & \multicolumn{3}{c||}{$\CSA$ language}  & \multicolumn{3}{c||}{$\NESA$ and $\SA$ language}          \\ \hline
class                                   & weak         & accept       & strong       & weak         & accept       & strong \\\hline
$\SMz(1)$                                 & $\checkmark$ & $\checkmark$ & $\checkmark$ & $\checkmark$ & $\checkmark$ & $\checkmark$ \\
between                                    & ?            &              &              & ?            & ?            & ?       \\
$\SMz(n)$               & ?            & $\checkmark$ & $\checkmark$ & ?            & ?            & ?        \\
between                                    &              &              &              &  *           & ?            & ?           \\
$\SMz(\infty)$                                  &              & ?            & $\checkmark$ &              & ?            & ?      \\\hline
 \end{tabular}
\caption{For each machine model $M$, with $\CSA$ in columns 2 through 4, and both $\NESA$ and $\SA$ in columns 5 through 7, for each of the space complexity measures $z$, weak, accept, or strong, and for each possible function represented in the rows, a $\checkmark$ is placed if there exists a language in the corresponding class, a blank space is left if there is no language in that class, and a ? is placed if it is open whether a language in that class exists. The entry marked * means that there does not exist such a language for stack automata, but it is open for non-erasing stack automata.}
\label{tab-results2}
\end{table}

\section{Decidability Properties Regarding Space Complexity of Stack Machines}

It is an easy observation that
when the space used by a checking stack automaton
is constant,
the device is no more powerful than a finite automaton.
Nevertheless,
given a checking stack automaton~$M$,
it is not possible to decide whether or not it accepts by using a constant amount of space with the weak measure.
This result can be derived by adapting
the argument used in~\cite{PP19}
for proving that,
when the weak measure is considered,
it is not decidable whether or not a nondeterministic pushdown automaton accepts by using a constant amount of pushdown store.
In that case,
the authors used a technique introduced in~\cite{Har67},
based on suitable encodings of single-tape Turing machine computations
and reducing the proof of the decidability to the halting problem; this can be done here as well.

\NOTE{Copied from appendix START}

\begin{proposition}
	It is undecidable whether a \CSA\ $M$ accepts in space~$\wspace(n)\in O(1)$ or not.
\end{proposition}
\begin{proof}	Given a deterministic Turing machine~$T$ with state set~$Q$, input alphabet $\Sigma$,
	and working alphabet~$\Gamma$,
	an encoding of a valid computation of~$T$ has the form~$\alpha_1\$\alpha_2\$\cdots\$\alpha_m$, $m\geq 1$,
	in which
	\begin{enumerate}[nosep]
		\item\label{it:configurations} $\alpha_i\in\Gamma^* Q \Gamma^*$ is the encoding of a valid configuration, for~$i=1,\ldots,m$;
		\item\label{it:initial} $\alpha_1$ is the encoding of the initial configuration of~$T$;
		\item\label{it:successive} $\alpha_{i+1}$ is the configuration reachable from~$\alpha_{i}$ by one step in $T$;
		\item\label{it:accepting} $\alpha_m$ is the encoding of an accepting configuration of~$T$.
	\end{enumerate}
	Encodings of partial valid computations of Turing machines are obtained by dropping Condition~\ref{it:accepting}.

	As proved by Hartmanis~\cite{Har67},
	the set of words encoding invalid computations of Turing machines is a context-free language.

	Let us consider a deterministic Turning machine~$T$ that, until it halts (if it does), uses
	arbitrarily large amounts of tape (a Turing machine can be made to make ``sweeps'' on the tape, always using a new cell on each sweep).

	We build a \CSA~$M_{T,x}$
	that accepts the complement of the set of partial valid computations of~$T$
	on a given input~$x$.

	The input is $u = \beta_1\$\beta_2\$\ldots \$\beta_r$, $r\geq 1, \beta_i \in \Gamma^*$.
	First, there is a regular language $R_1 = R_0(\$R_0)^*, R_0 = \Gamma^* Q \Gamma^*$ of all strings that satisfy condition \ref{it:configurations}. Similarly, if $\alpha_1$ is the initial configuration of $T$ on $x$, then there is a regular language $R_2 = \{\alpha_1\} \cup \{\alpha_1 \$ \alpha \mid \alpha \in (\Sigma\cup\{\$\})^*\}$ that satisfies condition \ref{it:initial}. Lastly, there is a language (not regular) $R_3$ of words satisfying condition \ref{it:successive}. The machine $M_{T,x}$ will accept $\overline{R_1} \cup \overline{R_2} \cup \overline{R_3}$, and is constructed as follows:
	At the beginning of the computation,
	$M_{T,x}$ guesses which one among Conditions~\ref{it:configurations},\ref{it:initial}, and~\ref{it:successive} will be verified to not be satisfied, and it performs the appropriate computation --- therefore each accepted word $u$ will have $\wspace(u)$ set to be the minimum possible under any of the three conditions. For all invalid computations $u \in \overline{R_1} \cup \overline{R_2}$, this can be verified entirely using the finite control (and not use the stack) since they are regular, and therefore $\wspace(u) = 0$ for all such words. To accept words in~$\overline{R_3}$, $M_{T,x}$ guesses some $\beta_i$, pushes $\beta_i$ to the stack, and verifies that $\beta_{i+1}$ cannot be produced from $\beta_i$ with one move of $T$. For all words $u \in (\overline{R_1} \cup \overline{R_2}) \cap \overline{R_3}$, $\wspace(u) = 0$. For those words $u \in \overline{R_3} \cap R_1 \cap R_2$, it can be higher. Assume that $u$ is in this set. Then, 
	$$\wspace(u) = \min_{1 \leq i \leq r-1} \{|\beta_i| \mid \beta_i \mbox{~does not produce~} \beta_{i+1}\mbox{~in one move of~} T\}.$$ 
	Let $j$ be the largest number such that $\beta_1 \$ \cdots \$ \beta_j$ is the encoding of the first~$j$ configurations of~$T$ on~$x$. We know $j \geq 1$ since~$u \in R_2$, and we know $j < r$ since $u \in \overline{R_3}$. Hence,
	$\wspace(u) \leq |\beta_j|$.
	
	Assume that $T$ accepts $x$ in some finite number of steps $m$. Let the accepting sequence of configurations of $T$ on $x$ be
	$v= \alpha_1\$ \cdots \$ \alpha_m$. Then $u$ must satisfy $\wspace(u) \leq |\alpha_m|$. Hence, every word in $\overline{R_1} \cup \overline{R_2} \cup \overline{R_3}$ requires at most $|\alpha_m|$ space, and $\wspace_{M_{T,x}}(n) \in \bigO(1)$.
	
	Assume that $T$ does not accept $x$ in some finite number of steps. For any arbitrarily large integer $s$, consider the string $ v = \alpha_1 \$ \cdots \$ \alpha_i \$ \beta$, where $|\alpha_i| \geq s$ (which must exist since $T$ will eventually increase the tape size past size $s$ by the assumption that it uses arbitrarily much space), and $\beta \in \Gamma^* Q \Gamma^*$ with $\alpha_i$ not producing $\beta$ in one move of $T$. In this case, $\wspace(v) \geq s$, and thus $\wspace_{M_{T,x}}(n) \notin \bigO(1)$.
	
	Hence, it is undecidable whether or not a given $\CSA$ $M$ satisfies $\wspace_M \in \bigO(1)$.
\end{proof}
It therefore follows that this property is undecidable for $\NESA$ and $\SA$.

\NOTE{Copied from appendis END}

On the other hand,
although it may seem counterintuitive,
the same problem is decidable for the accept and strong measures,
even for stack automata.
\begin{proposition}
\label{decidableconstant}
	For $z \in \{ {\rm a}, {\rm s} \}$, it is decidable whether an \SA\ $M$ satisfies $\zspace(M) \in \bigO(1)$ or not.
\end{proposition}
\begin{proof}
For the accept measure, first, we construct a finite automaton $M'$ accepting the store language of $M$. We can then decide finiteness of $L(M')$ since it is regular, which is finite if and only if $M$ operates in constant space.

For the strong measure, we can take $M$, and change it so that all states are final, then calculate the store language, and decide finiteness.
\MAR{Decidable whether a, w are bounded, use store language, check if finite.}
\end{proof}
Hence, this property is decidable for $\NESA$ and $\CSA$ as well.

\bigskip
Next, we consider the decidability of whether a $\CSA$ is ${\rm a}$-limited or ${\rm s}$-limited.
\begin{proposition}
\label{decidablelimited}
For $z \in \{ {\rm a}, {\rm s} \}$, it is decidable whether an \CSA\ $M$ is $z$-limited or not.
\end{proposition}
\begin{proof}
First consider the case for the accept measure.
Let $M$ be a $\CSA$. By the proof of Proposition \ref{acceptupperbound}, there exists a constant $\kappa$ such that, if there exists a critical section of length more than $\kappa$, then $M$ is ${\rm a}$-unlimited. Conversely, if all critical sections are of length at most $\kappa$ for some $\kappa$, then $M$ is ${\rm a}$-limited. From $M$, create $M'$ that simulates $M$ but at some non-deterministically guessed point, where it switches to pushing primed symbols to the stack only while simulating $\lambda$-transitions which it keeps using until some
non-deterministically guessed point, where it switches back to using unprimed symbols. In the read phase, it only allows $\lambda$-transitions while reading the primed symbols of the stack. Then it calculates the store language of $M'$ and it then erases the non-primed symbols with a homomorphism. This resulting language is finite if and only if $M$ is ${\rm a}$-limited. 

The proof also works for strong measure by making all states final.
\end{proof}
This problem is open for $\NESA$ and $\SA$.
\begin{corollary}
Let $M$ be a $\CSA$. For $z \in \{{\rm a}, {\rm s}\}$, exactly one of the following must occur.
\begin{enumerate}[nosep]
\item $M$ is $z$-unlimited;
\item $M$ is $z$-limited, $\zspace(n) \in O(1)$;
\item $M$ is $z$-limited, $\zspace(n) \in O(n), \zspace \notin o(n)$, and $\zlspace(n) \in \Theta(n)$.
\end{enumerate}
Furthermore, it is decidable, given $M$, which of these three cases occur.
\end{corollary}
\begin{proof}
The first statement is true by Theorem \ref{thm-csa-full}. Thus, it is first possible to test whether a given $M$ is $z$-unlimited by 
Proposition \ref{decidablelimited}, and if not, to decide if $\zspace(n) \in O(1)$
by Proposition \ref{decidableconstant}. If both are not true, then necessarily
$\zspace(n) \in O(n), \zspace \notin o(n)$, and $\zlspace(n) \in \Theta(n)$.
\end{proof}

\section{Conclusions and Future Directions}
In this paper, we defined and studied the weak, accept, and strong space complexity measures for variants of stack automata.
For checking stack automata with the accept or strong measures, there is ``gap'', and no function is possible between constant and linear, or above linear. For non-erasing stack automata, there are machines with complexity between constant and linear. 
Then, it is shown that for the strong measure, there is a checking stack language such that every machine accepting it is ${\rm s}$-unlimited (there is no function bounding the strong space complexity). 
Lastly, it is shown that it is undecidable whether a checking stack automaton has constant space complexity with the weak measure. But, this is decidable for both the accept and strong measures even for stack automata.

Many open problems remain, which is evident from the large number of ? symbols in Tables \ref{tab-results1} and \ref{tab-results2}. It is desirable to know whether there are any gaps between constant and $\bigO(\sqrt{n})$ and between that and linear space for the weak space complexity measure for checking stacks. It is open whether all stack automata have at most linear weak space complexity. The exact accept and strong space complexity functions possible for non-erasing and stack automata (besides constant, square root, and linear) still need to be determined. It is also open whether there is some stack language (or non-erasing stack language) such that every machine accepting it is ${\rm s}$-unlimited. Furthermore, for the accept measure, we conjecture that there is a $\CSA$ language whereby every machine accepting it is ${\rm a}$-unlimited.
Answering these open questions would be of interest to the automata theory community.



\bibliographystyle{ws-ijfcs}
\bibliography{space_of_stacks}{}
\end{document}